\documentclass[conference]{IEEEtran}
\IEEEoverridecommandlockouts

\usepackage{algorithm}
\usepackage{algpseudocode}
\usepackage{enumitem}

\algrenewcommand\algorithmicprocedure{\small \textbf{\textsf{procedure}}}

\algrenewtext{Procedure}[2]%
{\algorithmicprocedure\ \normalsize \textsc{\textrm{#1}}#2}

\algnewcommand\And{\textbf{and} }

\usepackage[dvipsnames]{xcolor}
\definecolor{mtlbgreen}{RGB}{34, 139, 34}

\algrenewcommand{\algorithmiccomment}[1]{\small\sffamily\textcolor{mtlbgreen}{\textit{/ #1}}} 

\makeatletter
\newcommand{\removelatexerror}{\let\@latex@error\@gobble}
\def\NAT@spacechar{~}
\makeatother

\usepackage{bbm}
\usepackage{bm}
\usepackage{amssymb}
\usepackage{amsthm}
\usepackage{tikz}
\usepackage{mleftright}
\usetikzlibrary{math}
\usetikzlibrary{shapes}
\usetikzlibrary{positioning}

\usepackage[utf8]{inputenc} 
\usepackage[T1]{fontenc}
\usepackage{url}
\usepackage{ifthen}
\usepackage[noadjust]{cite}
\usepackage[cmex10]{amsmath}
\usepackage{amssymb}
\usepackage{amsthm}

\usepackage{pgfplots}

\usepackage{euscript}
\DeclareMathAlphabet\EuRoman{U}{eur}{m}{n}
\SetMathAlphabet\EuRoman{bold}{U}{eur}{b}{n}
\newcommand{\eurom}{\EuRoman}

\usepackage[english]{babel}

\newtheorem{thm}{Theorem}
\newtheorem{cor}{Corollary}
\newtheorem{lem}{Lemma}

\newcommand{\C}{\mathsf{C}}

\newcommand{\I}{\mathsf{I}}
\newcommand{\mi}[2]{{\I}\mleft(#1 \, ; #2 \mright)}

\newcommand{\HH}{\mathsf H}
\newcommand{\h}{\mathsf h}
\newcommand{\ent}[1]{{\h}\mleft(#1\mright)}

\newcommand{\entcnd}[2]{{\h}\mleft(\mleft. #1 \,\mright| #2\mright)}

\newcommand{\Exp}{\mathsf E}
\newcommand{\expect}[1]{{\Exp}\mleft[#1\mright]}
\newcommand{\expcnd}[2]{{\Exp}\mleft[\mleft. #1 \,\mright|\, #2\mright]}

\newcommand{\Gammaf}[1]{\Gamma \mleft( #1 \mright)}

\newcommand{\X}{{\bf X}}

\newcommand{\Y}{{\bf Y}}

\newcommand{\RR}{\mathbb{R}}

\newcommand{\Yv}{{\bf Y}}
\newcommand{\Xv}{{\bf X}}
\newcommand{\xv}{{\bf x}}

\newcommand{\Wv}{{\bf W}}
\newcommand{\Iv}{\textbf{I}}

\newcommand{\lr}[1]{\mleft( #1 \mright)}
\newcommand{\lrs}[1]{\mleft[ #1 \mright]}
\newcommand{\lrc}[1]{\mleft\{ #1 \mright\}}

\newcommand{\N}{\eurom{N}}
\newcommand{\A}{\eurom{A}}
\newcommand{\K}{\eurom{K}}
\newcommand{\Ss}{\eurom{S}}
\newcommand{\Ell}{\mathcal{L}}
\pgfplotsset{compat=1.17}

\begin{document}
 
\title{The Capacity of the Amplitude-Constrained Vector Gaussian Channel}

\author{%
   \IEEEauthorblockN{Antonino Favano\IEEEauthorrefmark{1}\IEEEauthorrefmark{2},
                     Marco Ferrari\IEEEauthorrefmark{2},
                      Maurizio Magarini\IEEEauthorrefmark{1},
                     and Luca Barletta\IEEEauthorrefmark{1}}
   \IEEEauthorblockA{\IEEEauthorrefmark{1}%
                     Politecnico di Milano,
                     Milano, Italy,
                     \{antonino.favano, maurizio.magarini, luca.barletta\}@polimi.it}
   \IEEEauthorblockA{\IEEEauthorrefmark{2}%
                     CNR-IEIIT, Milano, Italy,
                     marco.ferrari@ieiit.cnr.it}
 }

\maketitle

\begin{abstract}
The capacity of multiple-input multiple-output additive white Gaussian noise channels is investigated under peak amplitude constraints on the norm of the input vector. New insights on the capacity-achieving input distribution are presented. Furthermore, it is provided an iterative algorithm to numerically evaluate both the information capacity and the optimal input distribution of such channel.
\end{abstract}
\vspace*{-.1cm}

\section{Introduction}
The channel capacity of multiple-input multiple-output (MIMO) additive white Gaussian noise (AWGN) channels under peak amplitude constraints is of great interest for wireless communication systems, since such systems are typically subject to peak power constraints to cope with the nonlinear nature of power amplifiers. In~\cite{Smith}, the author proves that, for the scalar Gaussian channel, the capacity-achieving input probability distribution is discrete with a finite number of mass points. A similar result is extended to quadrature Gaussian channels in~\cite{Shamai}, which proves that the optimal input distribution is discrete in amplitude, with a finite number of mass points, and independent uniformly distributed phase. In~\cite{Rassouli}, the authors generalize the mentioned result for vector Gaussian channels, by proving that the capacity-achieving input distribution has probability masses assigned to a finite discrete set of amplitude values, each uniformly distributed on the hyper-sphere with the corresponding radius. Although the structure of the optimal input distribution is known, the evaluation of the channel capacity is nontrivial and requires to determine the optimal set of amplitude values and the associated probability masses. An upper bound on the number of mass points of the optimal input distribution is presented in~\cite{DytsoShells}. Moreover, capacity bounds for MIMO channels under constrained norm of the input vector are derived in~\cite{thangaraj2017capacity}.

In the mentioned works, the MIMO channel matrix is assumed to be an identity matrix. More general results, valid for arbitrary MIMO channel matrices and peak-amplitude constraints are presented in~\cite{Dytso} and~\cite{ourISIT2020}.
\vspace*{-.1cm}

\subsection*{Contribution}
We investigate the optimal input distribution and the information capacity of an AWGN MIMO system under a peak-amplitude constraint on the norm of the input vector. We show some properties of the capacity-achieving distribution and we propose a lower bound on its number of mass points. Finally, we numerically derive the optimal input distribution and the channel capacity via the combination of two nested iterative algorithms, which are variants of the Blahut-Arimoto and the Dynamic Assignment Blahut-Arimoto algorithms~\cite{Farsad}.
\vspace*{-.1cm}

\subsection*{Paper Organization}
In Sec.~\ref{S:PrevWorks}, we summarize the main results available prior to this paper. In Sec.~\ref{S:sysmod}, we define the system model and then, in Sec.~\ref{S:nAMP}, we provide new insights on the support of the capacity-achieving distribution. Section~\ref{S:ShellsC} introduces the iterative algorithms used to numerically evaluate the channel capacity and the associated input distribution. As a case study, we present numerical results for a $2 \times 2$ complex-valued MIMO system. Finally, Sec.~\ref{S:conclusion} concludes the paper.
\vspace*{-.1cm}

\subsection*{Notation}
We use bold letters for vectors ($\xv$) and uppercase letters for random variables ($X$). The amplitude (Euclidean norm) and direction of $\xv$ are denoted by $|\xv|$ and $\angle \xv \triangleq \xv/|\xv|$, respectively. Given a random variable $X$, its probability density function (PDF), mass function (PMF), and cumulative distribution function are denoted by $f_X$, $P_X$, and $F_X$, respectively. We denote by $\I$ the mutual information, by $\h$ the differential entropy, and by $\HH$ the entropy. Moreover, $I_n(\cdot)$ is the modified Bessel function of the first kind of order $n$. We represent the $n \times 1$ vector of zeros by $\textbf{0}_n$ and the $n \times n$ identity matrix by $\Iv_n$. We denote by ${\cal CN}(\boldsymbol{\mu},\mathsf{\Sigma})$ a multivariate proper complex Gaussian distribution with mean vector $\boldsymbol{\mu}$ and covariance matrix $\mathsf{\Sigma}$, while ${\cal U}(a,b)$ is the scalar uniform distribution over the interval $(a,b)$. We define $\chi^2_{k}(\xi)$ as the noncentral chi-squared distribution with $k$ degrees of freedom and with noncentrality parameter $\xi$, and we indicate its PDF with $f_{\chi^2_{k}(\xi)}$.

\section{Previous Works} \label{S:PrevWorks}
By using the dual capacity expression, the authors of~\cite{thangaraj2017capacity} provide a capacity upper bound for the AWGN MIMO channel subject to a constraint on the norm of the input vector. The upper bound can be evaluated analytically up to a certain signal-to-noise ratio (SNR) level, and computed numerically otherwise. For the complex AWGN channel, a lower bound close to the upper bound is provided for any SNR value. Instead, for $n \times n$ AWGN MIMO channels with $n>1$ they propose a capacity lower bound based on the entropy power inequality (EPI). Although being small in the high SNR regime, the capacity gap between upper bound and lower bound widens at intermediate SNR. Moreover, those bounds do not provide definite insights on the optimal input distribution at intermediate SNR. The authors of~\cite{DytsoShells} provide a lower bound on the number of amplitude values for scalar Gaussian channels and upper bounds for both the scalar and vector cases. Finally, important insights on the optimal input distribution in the low SNR regime are presented in~\cite{Dytso1shell}.

\section{Channel Model} \label{S:sysmod}
Let us consider an $\N \times \N$ complex AWGN MIMO system. The input-output relationship is defined by
\begin{align}
    \Yv = \Xv + \Wv,
\end{align}
where $\Yv$$\,\in\,$$\mathbb{C}^{\N}$ and  $\Xv$$\,\in\,$$\mathbb{C}^{\N}$ are the output and input vectors, respectively, and the noise vector $\Wv$$\,\in\,$$\mathbb{C}^{\N}$ is such that $\Wv$$\,\sim\,$${\cal CN}(\textbf{0}_{\N}, 2\Iv_{\N})$. The peak amplitude constraint on the norm of the input vector is given by $|\Xv|$$\,\le\,$$ \A$, for $\A \in \RR^+$. The MIMO channel capacity is then
\begin{align} 
    \C(\A) &\triangleq \max_{F_\Xv: \: |\X|\le \A } \mi{\Xv}{\Yv} \\
    &= \max_{F_\Xv: \: |\X| \le \A } \ent{\Yv} - \ent{\Wv}, \label{eq:capacity}
\end{align}
where $F_\Xv$ is the input probability law. 

By symmetry of the problem, the optimal input and output distributions are isotropic. Therefore, we can transform the maximization over $F_\X$ in~\eqref{eq:capacity} into a mono-dimensional problem, dependent only on the amplitude $|\X|$. We notice that
\begin{align}
    \ent{\Yv} &\stackrel{(a)}{=} \ent{|\Yv|} + (2\N - 1) \expect{\log |\Yv|} + \h_\lambda \lr{\angle{\Yv}} \\
        &\stackrel{(b)}{=} \ent{|\Yv|^2} + (\N - 1) \expect{\log |\Yv|^2} + \log \frac{\pi^\N}{\Gammaf{\N}}, \label{eq:absYent}
\end{align}
where step $(a)$ holds by~\mbox{\cite[Lemma~6.17]{LapidothMoser}} and by independence between $|\Yv|$ and $\angle{\Yv}$, $\h_\lambda \lr{\cdot}$ is a differential \mbox{entropy-like} quantity for random vectors on the unit sphere in $\mathbb{C}^\N$~\mbox{\cite[Lemma~6.16]{LapidothMoser}}, and step $(b)$ holds because $\angle{\Yv}$ is uniform and thanks to~\mbox{\cite[Lemma~6.15]{LapidothMoser}}. By defining
\begin{align}
    \ell(F_\Xv) \triangleq \ent{|\Yv|^2} + (\N - 1) \expect{\log |\Yv|^2},
\end{align}
we can rewrite the capacity as
\begin{align}
  	\C(\A)&=\max_{F_\X :|\X|\le \A} \ell(F_{\X}) + \log\frac{\pi^\N}{\Gamma(\N)}-\ent{\Wv} \\
	&=\max_{F_\X :|\X|\le \A} \ell(F_{\X}) - \log((2e)^\N\Gamma(\N)). \label{eq:finalC}
\end{align}
Furthermore, we remark that given
\begin{align}
  |\Yv|^2 & =  \sum_{i=1}^{\N}  |X_i  +  W_i|^2  \stackrel{d}{=} \big| |\Xv| +  W_1 \big|^2 + \sum_{i=2}^{\N}  |W_i|^2,
\end{align}
where $\stackrel{d}{=}$ means equality in distribution, each component of the output vector is $(Y_i | X_i = x_i) \sim {\cal CN}(x_i,2) $, therefore $(|\Yv|^2 \hspace{2pt} | \hspace{2pt} \Xv = \xv ) \sim \chi^2_{2\N}(\xi)$ with \mbox{$\xi = \sum_{i=1}^{\N} |x_i|^2 = |\xv|^2$}.
From~\cite{Rassouli}, we know that the optimal $F_\Xv$ is composed of a finite number of hyper-spheres centered in the origin. We assign the probability $p_i$ to the $i$th sphere with radius $\rho_i \le \A$. Since $\Xv$ is isotropically distributed, the probability density of a point on the sphere of radius $\rho_i$ is $p_i/\Ss_i$ where $\Ss_i \triangleq 2 \pi^\N \rho_i^{2\N-1} / \Gammaf{\N}$ is the hyper-surface measure for the $i$th sphere. Therefore, the PDF of the optimal output $|\Yv^\star|^2$ takes the form
\begin{align} \label{eq:absYpdf}
    f_{|\Yv^\star|^2}(y) = \sum_{i=1}^\K p_i f_{\chi^2_{2\N}(\rho^2_i)}(y), \qquad y>0.
\end{align}
By plugging~\eqref{eq:absYpdf} into~\eqref{eq:finalC}, we notice that the maximization can now be carried out just over any \mbox{mono-dimensional} PMF $P_{|\Xv|}$ defined as
\begin{align}
    P_{|\Xv|}(\rho_i) \triangleq p_i, \ \forall i=1,\dots,\K.
\end{align}
Finally, we define the \mbox{$\text{SNR} \triangleq \A^2 / (2 \N )$}.
In the following section, we provide new insights on the capacity-achieving input distribution $P_{|\X^\star|}$.

\section{New Insights on the Optimal Input PMF} \label{S:nAMP}
Although it is not simple to evaluate the optimal number of points in the input PMF and the corresponding amplitude values, there are some special cases that are more tractable and provide some intuitions about the general \emph{shape} of the optimal input PMF. \mbox{In~\cite[Thm.~2]{Dytso1shell}}, the authors prove that at low SNR (i.e., $\A<\bar{\A}_N)$ the optimal input is composed of a single sphere of maximum radius $\A$. The value of $\bar{\A}_N$ is determined by numerically solving an integral equation. 
\subsection{Lower Bound on the Optimal Number of Spheres} 
Let us denote by $\K^\star$ the number of mass points in $P_{|\Xv^\star|}$.
\begin{thm}
	A lower bound on the number of mass points $\K^\star$ for $P_{|\Xv^\star|}$ is given by
	\begin{equation}
	\K^\star \ge \underline{\K} \triangleq \mleft\lceil \sqrt{\frac{\mleft(\A^2+2 e\mright)^2 + 8\pi (\N-1)}{8\pi e (\N+ \A^2/2)}} \ \mright\rceil.
	\end{equation}
\end{thm}
\begin{proof}
	We have
	\begin{align}
	\log \K^\star = \log|\mathsf{supp}\lr{P_{|\X^\star |}}| &= \log|\mathsf{supp}\lr{P_{|\X^\star |^2}}| \\
	&\ge \HH\lr{|\X^\star|^2} \\
	&\ge \mi{|\X^\star|^2}{|\Y^\star|^2} \\
	&\ge \mi{|\X|^2}{|\Y|^2}, 
	\label{eq:mi_squared}
	\end{align}
	where the first two inequalities hold because $|\X^\star|$ is discrete and the last inequality from the suboptimal choice $|\X|^2 \sim {\cal U}[0,\A^2]$. Let us lower-bound the mutual information in \eqref{eq:mi_squared}: we have
	\begin{align}
	&\ent{|\Y|^2} = \ent{\big||\X|+W_1\big|^2+\sum_{i=2}^{\N} |W_i|^2} \\
	&\stackrel{(a)}{\ge} \frac{1}{2}  \log\mleft(\hspace{-0.6ex}\exp\mleft(\hspace{-0.3ex}2\ent{\big| |\X| + W_1\big|^2}\mright)\hspace{-.07cm}+\hspace{-.03cm}\exp\mleft(\hspace{-0.6ex}2\ent{\sum_{i=2}^{\N} |W_i|^2}\hspace{-.1cm}\mright) \hspace{-.1cm}\mright) \label{eq:epi} \\
	&\stackrel{(b)}{\ge} \frac{1}{2}\log \mleft(\exp\hspace{-.02cm}\mleft(2 \log\mleft(\A^2\hspace{-.05cm}+\hspace{-.05cm}2 e\mright)\mright)\hspace{-.02cm} + \exp\mleft(\log(8\pi (\N\hspace{-.07cm}-\hspace{-.07cm}1)\hspace{-.01cm})\hspace{-.02cm}\mright)\mright) \label{eq:entY}\\
	&= \frac{1}{2}\log\mleft(\mleft(\A^2+2 e\mright)^2 + 8\pi (\N-1)\mright), \label{eq:ent_unif}
	\end{align}
	where we used the EPI in step $(a)$, a lower bound on the entropy of a $\chi^2_{2\N-2}$ variate~\mbox{\cite[App.~C]{barletta2020capacity}} in step $(b)$ together with the following bound
	\begin{align}
	&\ent{\big| |\X| + W_1 \big|^2} = \ent{|\X|e^{j\Phi}+W_1} - \log(\pi)\\
	&\ge \log\mleft(\exp\mleft(\ent{|\X|e^{j\Phi}}\mright) + \exp\mleft(\ent{W_1}\mright)\mright) -\log(\pi) \\
	&= \log\mleft(\pi \exp\mleft(\ent{|\X|^2}\mright) + 2\pi e\mright) - \log(\pi) \\
	&= \log\mleft(\exp\mleft(\ent{|\X|^2}\mright) + 2 e\mright) \\
	&= \log\mleft(\A^2 + 2 e\mright),
	\end{align}
	where $\Phi \sim {\cal U}[0,2\pi)$ is independent of any other quantity, and in the last step we used $\ent{|\X|^2}=\log(\A^2)$.
	The conditional entropy is
	\begin{align}
	\entcnd{|\Y|^2}{|\X|^2} &\le \frac{1}{2}\expect{\log(8\pi e (\N+|\X|^2))} \label{eq:thm non central}\\
	&\le \frac{1}{2}\log \lr{8\pi e \lr{\N+\expect{|\X|^2}}} \label{eq:jensen}\\
	&= \frac{1}{2}\log(8\pi e (\N+ \A^2/2))\label{eq:entYgivenX},
	\end{align}
	where, in order, we have used the Gaussian maximum entropy principle, Jensen's inequality, and finally $\expect{|\X|^2}= \A^2/2$. Putting~\eqref{eq:mi_squared}, \eqref{eq:ent_unif}, and~\eqref{eq:entYgivenX} together, proves the claim.
\end{proof}

\subsection{On the Optimality of the Outer Sphere} \label{sS:point in A}

We prove that at any SNR level the sphere of maximum radius $\A$ is always part of the optimal input PMF. Let $i(\cdot \ ;F_{|\Xv|})$ be the information density with respect to the norm of the input vector
\begin{align} \label{eq:infdensity}
    i(\rho;F_{|\Xv|}) \triangleq &     \int_0^{\infty}\hspace{-.35cm} f_{\chi^2_{2\N}(\rho^2)}(\hspace{-.02cm}y\hspace{-.02cm}) \hspace{-.02cm}\log\hspace{-.02cm} \frac{y^{\N-1}}{\int_0^\A\hspace{-.15cm} f_{\chi^2_{2\N}(t^2)}(\hspace{-.02cm}y\hspace{-.02cm})dF_{|\Xv|}(t)}dy \hspace{-.05cm} \nonumber \\ &- \log\lr{(2e)^\N \Gamma( \N ) }.
\end{align}
\begin{lem}\label{Lem:derivative}
	The derivative of the information density is
	\begin{align}
	&i'(\rho;F_{|\X|}) =-2\rho\: \Exp \Bigg[ \Exp \Bigg[ \frac{1}{2}\frac{|\X|}{|\Y|} \frac{I_{\N-2}(|\X| |\Y|)}{I_{\N-1}(|\X| |\Y|)} \nonumber \\
	& \hphantom{i'(\rho;F_{\X}) =-2\rho\: \Exp \Bigg[ \Exp} - \lr{ \frac{1}{2}+\frac{\N-1}{|\Y|^2} } \bigg| |\Y|^2=Q' \Bigg] \Bigg]
	\end{align}
	where $Q'\sim \chi^2_{2(\N+1)}(\rho^2)$.
\end{lem}
\begin{proof}
See the Appendix.
\end{proof}

\begin{thm}\label{Thm:outer_shell}
	Consider the constraint $|\X|\le \A$ and an input such that $P(|\X|>c)=0$, with $c<\A$. Then, the information density is strictly increasing for $\rho>c$, i.e.,
	\begin{equation}
	i'(\rho; F_{\X})\ge \frac{\rho}{1+\frac{c^2(2\N+\rho^2)}{4(\N-1/2)^2}}>0, \qquad \rho>c.
	\end{equation}
\end{thm}
\begin{proof}
	We introduce the function
	\begin{equation}\label{eq:dec_s}
	s_{\N-1}(t) = \frac{1}{t} \frac{I_{\N-1}(t)}{I_{\N-2}(t)}, \qquad t>0.
	\end{equation}
	As shown in the proof of~\mbox{\cite[Lemma 2]{yu2011log}}, the function $s_{\N-1}$ is decreasing in $t$ for any $\N>3/2$, hence its derivative is negative. This is also true for $\N=1$. To prove it, we remark that $I_{-1}(t)$$\,=\,$$I_1(t)$ and, therefore, the derivative of $s_0(t)$ can be written as
	\begin{align} \label{eq:s0derivative}
	    s'_0(t) \hspace{-.05cm}=\hspace{-.05cm} \frac{1}{t}\lr{1\hspace{-.05cm}-\hspace{-.05cm}\frac{I_0^2(t)}{I_1^2(t)}} \hspace{-.05cm}=\hspace{-.05cm} \frac{1}{t} \lr{1\hspace{-.05cm}+\hspace{-.05cm}\frac{I_0(t)}{I_1(t)}} \lr{1\hspace{-.05cm}-\hspace{-.05cm}\frac{I_0(t)}{I_1(t)}}.
	\end{align}
    Since $I_1(t)<I_0(t)$, for $t\ge 0$, with \eqref{eq:s0derivative} we showed that $s'_0(t)<0$ for all $t>0$. 
    
    Using \eqref{eq:dec_s}, \eqref{eq:s0derivative}, and the fact that $P(|\X|>c)=0$, we can lower-bound the derivative of the information density as 
	\begin{align*}
	i'(\rho;F_{|\X|}) \ge -2\rho \: \expect{\frac{1}{2}\frac{c}{|\Y|} \frac{I_{\N-2}(c |\Y|)}{I_{\N-1}(c |\Y|)}-\frac{1}{2}-\frac{\N-1}{|\Y|^2} },
	\end{align*}
	where $|\Y|^2 \sim \chi^2_{2(\N+1)}(\rho^2)$. In the following, we use the relationship~\mbox{\cite[Eq.~9.6.26]{abramowitz1988}}
	\begin{equation}\label{eq:abramo}
	I_{\nu-1}(z)-I_{\nu+1}(z)=\frac{2\nu}{z}I_\nu(z).
	\end{equation}
	By assigning $\nu=\N-1$ and $z=c|\Y|$, we get
	\begin{align}
	&i'(\rho;F_{|\X|}) \nonumber \\ 
	& \hspace{-.05cm}\ge \hspace{-.07cm}-2\rho \: \expect{\frac{\N-1}{|\Y|^2} \frac{I_{\N}(c |\Y|)}{I_{\N-2}(c |\Y|)-I_{\N}(c |\Y|)}-\frac{1}{2}} \\
	&\hspace{-.05cm}\stackrel{(a)}{=}\hspace{-.07cm}-2\rho \: \expect{\frac{c}{2|\Y|} \frac{I_{\N}(c |\Y|)}{I_{\N-1}(c |\Y|)}-\frac{1}{2}} \\
	&\hspace{-.05cm}=\hspace{-.07cm}-\rho \mleft(\int_{0}^{\infty}\hspace{-.2cm}\frac{c}{2\sqrt{t}}e^{-\frac{t+\rho^2}{2}} \frac{t^\frac{\N}{2}}{\rho^\N} I_N(\rho\sqrt{t}) \frac{I_{\N}(c \sqrt{t})}{I_{\N-1}(c \sqrt{t})}dt\hspace{-.05cm}-\hspace{-.07cm}1\mright) \\
	&\hspace{-.05cm}=\hspace{-.07cm}-\rho\: \expect{\frac{c}{\rho}\frac{I_{\N}(c \sqrt{T})}{I_{\N-1}(c \sqrt{T})}\frac{I_{\N}(\rho \sqrt{T})}{I_{\N-1}(\rho \sqrt{T})}-1}, \label{eq:der}
	\end{align}
	where in step $(a)$ we used again \eqref{eq:abramo} and in the last equality we rearranged the integral by using $T\sim \chi_{2\N}^2(\rho^2)$. We introduce the inequality from~\mbox{\cite[Eq.~(16)]{amos1974}}
	\begin{equation} \label{eq:amosineq}
	\frac{I_{\nu+1}(z)}{I_{\nu}(z)}\le R_{\nu+1}(z) \triangleq \frac{z}{\nu+1/2+\sqrt{(\nu+1/2)^2+z^2}},
	\end{equation}
	where $\nu\ge 0, \, z\ge 0$. By plugging~\eqref{eq:amosineq}
	into \eqref{eq:der}, we get
	\begin{align}
	i'(\rho;F_{|\X|}) &\ge -\rho\: \expect{cR_N\lr{c\sqrt{T}}\frac{1}{\rho}R_N\lr{\rho\sqrt{T}}-1} \\
	&\stackrel{(a)}{\ge} -\rho\: \expect{R_N\lr{c\sqrt{T}}R_N\lr{c\sqrt{T}}-1} \\
	&\stackrel{(b)}{\ge} -\rho\: \expect{\frac{c^2 T}{4(\N-\frac{1}{2})^2+c^2T}-1} \\
	&=\expect{\frac{\rho}{1+\frac{c^2T}{4(\N-\frac{1}{2})^2}}} \\
	&\stackrel{(c)}{\ge} \frac{\rho}{1+\frac{c^2\expect{T}}{4(\N-\frac{1}{2})^2}} =\frac{\rho}{1+\frac{c^2(2\N+\rho^2)}{4(\N-\frac{1}{2})^2}},
	\end{align}
	where in step $(a)$ we used $\rho>c$, in step $(b)$ we used  $T\ge 0$ in the denominator, in step $(c)$ Jensen's inequality, and, finally, $\expect{T}=2\N+\rho^2$.
\end{proof}

As an immediate consequence of Thm.~\ref{Thm:outer_shell}, the Karush– Kuhn–Tucker conditions~\mbox{\cite[Lemma~7]{DytsoShells}}, and the fact that the optimal distribution is discrete, we have the following:
\begin{cor} \label{cor: rho_A}
	The sphere with radius $\rho=\A$ is always part of the capacity achieving distribution, i.e., $P_{|\X^\star |}(\A)>0$.
\end{cor}

\section{Iterative Evaluation of the Capacity} \label{S:ShellsC}
The numerical evaluation of the channel capacity comprises two nested iterative procedures. 
Given a PMF $P_{|\Xv |}$ with a fixed number of spheres $\K$ and their radii $\{\rho_i\}$, the inner layer iteratively updates the probabilities $\{ p_i \}$ until the corresponding mutual information is maximized. Since the tentative $\K$ and $\{ \rho_i \}$ might not be optimal, the resulting mutual information is a lower bound on the channel capacity.

The outer layer iteratively updates $\K$ and $\{ \rho_i \}$, runs the inner layer, and computes an upper bound on the channel capacity relying on the dual capacity expression based on the informational divergence~\cite{csiszar2011information}. The iterative algorithm of the outer layer continues until the lower bound from the inner layer and the upper bound converge within a given tolerance.

The inner layer is a variant of the Blahut-Arimoto algorithm: compared to its original version, it allows for a continuous output distribution and a continuous scalar input search space. The outer layer is analogous to the dynamic assignment Blahut-Arimoto algorithm of~\cite{Farsad}, using the mentioned variant of the Blahut-Arimoto algorithm.

\subsection{Inner Layer: Probabilities Update}
Given the number of spheres $\K$, their radii $\{\rho_i\}$, and probabilities $\{p_i\}$, we have
\begin{align} 
\ell(P_{|\Xv|}) = \sum_{i=1}^{\K} p_i \int_{0}^{\infty} \hspace{-0.25cm} f_{\chi_{2\N}^2(\rho^2_i)}(y) \log\mleft(\frac{y^{\N-1}}{r(y,\{p_j\})}\mright) dy, \label{eq:itercapacity}
\end{align}
where 
\begin{align}
r(y,\{p_j\}) = \sum_{j=1}^{\K} p_j f_{\chi_{2\N}^2(\rho^2_j)}(y).
\end{align}
Let us rewrite the maximization of~\eqref{eq:itercapacity} over $\{p_i\}$ as follows
\begin{align}
&\max_{\{p_i\}} \ell(P_{|\Xv|}) = \max_{\{p_i\}}  \sum_{i=1}^{\K} p_i \Bigg\{ \int_{0}^{\infty} f_{\chi_{2\N}^2(\rho^2_i)}(y) \nonumber \\
& \ \bigg[ \log \frac{p_i f_{\chi_{2\N}^2(\rho^2_i)}(y)y^{\N-1}}{r(y,\{p_j\})} -\log(f_{\chi_{2\N}^2(\rho^2_i)}(y))\bigg]dy -\log(p_i)\Bigg\} \\
&=\max_{\{p_i\}} \max_{r_i(y)}  \Ell(\{p_i\},\{r_i(y)\}) \label{eq:doublemax}
\end{align}
with
\begin{align}
&\Ell(\{p_i\},\{r_i(y)\}) \triangleq \nonumber \\
& \ \sum_{i=1}^{\K} p_i \Bigg\{\int_{0}^{\infty} f_{\chi_{2\N}^2(\rho^2_i)}(y) \log\frac{ r_i(y) y^{\N-1}}{f_{\chi_{2\N}^2(\rho^2_i)}(y)}dy -\log(p_i) \Bigg\},
\end{align}
where, for any given $y\in\mathbb{R}^{+}$, $\{r_i(y)\}_{i=1}^\K$ is a valid PMF. This is a test PMF that allows us to write the capacity as a nested maximization problem. Specifically, the last step of~\eqref{eq:doublemax} holds thanks to the nonnegativity of the \mbox{Kullback-Leibler} divergence between $\{ p_i f_{\chi_{2\N}^2(\rho^2_i)}(y)/ r(y,\{p_j\}) \}$ and $\{ r_i(y) \}$.

We leverage the double maximization to devise the iterative procedure described in~Algorithm~\ref{alg:BA-like}.
\begin{figure}[!t]
 \removelatexerror
  \begin{algorithm}[H]
  \caption{Inner Layer}
  \label{alg:BA-like}
\begin{algorithmic}[1]
    \Procedure{Algorithm1}{$\lr{P_{|\Xv|}}$} \vspace*{0.35ex}
    \State {$\varepsilon_1$} \hskip2ex \Comment{Desired tolerance in evaluating \eqref{eq:doublemax}}
    \State {$M$} \hskip2ex \Comment{Number of iterations within $\varepsilon_1$ for convergence}
    \State $q \leftarrow 0$ \hskip2ex \Comment{Iteration index} \vspace*{0.2ex}
    \State $\K \gets \mleft|\mathsf{supp} \lr{P_{|\Xv|}}\mright| $
    \Repeat
    \State $q \gets q + 1$ \vspace*{0.2ex}
    \State $p_i \gets P_{|\Xv|}\lr{\rho_i}, \ \forall i=1,\dots,\K$
    \For{$i \gets 1, \dots , \K$}
        \State \hspace*{-2ex} \Comment{Set optimal PMF of inner maximization  in~\eqref{eq:doublemax} as} \vspace*{-1ex}
        \State $\vcenter{\begin{flalign}
        \hspace*{-1.5ex} r_i(y) \gets \frac{p_i f_{\chi_{2\N}^2(\rho^2_i)}(y)}{\sum_{j=1}^{\K} p_j f_{\chi_{2\N}^2(\rho^2_j)}(y)}, \quad \forall y\in\mathbb{R}^{+}. && \hspace*{-12.6ex}
        \end{flalign}}$ \vspace*{-0.75ex}
        \State \hspace*{-2ex} \Comment{Set optimal PMF of outer maximization in~\eqref{eq:doublemax} as} \vspace*{-.75ex}
        \State $\vcenter{\begin{flalign*}
            \hspace*{-1.3ex} p'_i \gets  \exp \lr{ \int_{0}^{\infty} \hspace*{-.35cm} f_{\chi_{2\N}^2(\rho^2_i)}(y) \log\frac{r_i(y) y^{\N-1}}{f_{\chi_{2\N}^2(\rho^2_i)}(y)}dy }, &&
        \end{flalign*}}$ \vspace*{-.75ex}
    \EndFor
    \State $P_{|\Xv|}\lr{\rho_i} \gets p'_i/ \sum_{j=1}^\K p'_j, \quad \forall i=1,\ldots,\K$ \vspace*{.5ex}
    \State $l(q) \gets \Ell\lr{\{P_{|\Xv|}\lr{\rho_i} \},\{r_i(y)\}}$ \vspace*{0.5ex}
    \Until{$q>M \textsf{\textbf{ and }} |l(q)-l(q-m) | < \varepsilon_1, \ \forall m=1,\dots,M$} \vspace*{-0.75ex}
    \State \Comment{Evaluate the Lower Bound on Channel Capacity} \vspace*{-1ex}
    \State $\vcenter{\begin{flalign}
    \underline{\C}(\A) \gets l(q) - \log \lr{ (2e)^\N\Gamma(\N) }. && \hspace{-6.9ex}
    \end{flalign}}$ \vspace*{-1ex}
    \State \Return {$P_{|\Xv|},\underline{\C}(\A)$} \vspace*{0.3ex}
    \EndProcedure
\end{algorithmic}
\end{algorithm}
\vspace*{-6ex}
\end{figure}
\subsection{Outer Layer: PMF Evolution}
In Corollary~\ref{cor: rho_A}, we proved that $\rho_i = \A$ is always an optimal mass point and in~\mbox{\cite[Thm.~2]{Dytso1shell}} the authors proved that it is the only mass point at low SNR. As the SNR increases, a new mass point, hence a second sphere, can be added. 
In~\mbox{\cite[Thm.~4]{Dytso1shell}}, they prove that the second sphere arises in $\rho_2=0$.
Via numerical evaluation we notice that, by slightly increasing the SNR, $i(0;P_{|\Xv|})$ eventually grows larger than any $i( \rho;P_{|\Xv|})$. This behavior remains consistent for any $\K$. Therefore, we conjecture that for a sufficiently small increment in SNR, at most a single mass point is added to the support and that this mass point arises in $\rho_{\K+1}=0$. Moreover, we reduce the computational burden by exploiting the results of Sec.~\ref{S:nAMP}, i.e., $\K^\star \ge \underline{\K}$ and $\rho_1 = \A$. 
Let us consider $\K>1$. Given a tentative input PMF, the derivative of the mutual information with respect to $\rho^2_i$ is
\begin{align}
\frac{\partial}{\partial \rho^2_i} \I =& \int_{0}^{\infty} \frac{p_i}{2} \Big( f_{\chi_{2\N+2}^2(\rho^2_i)}(y) - f_{\chi_{2\N}^2(\rho^2_i)}(y) \Big) \nonumber \\
&\lr{\log\mleft(r_i(y) y^{\N-1}\mright)-\log \lr{f_{\chi_{2\N}^2(\rho^2_i)}(y)}}dy. \label{eq:derivatives}
\end{align}
\begin{figure}[!t]
 \removelatexerror
  \begin{algorithm}[H] 
  \caption{Outer Layer}
  \label{alg:2}
\begin{algorithmic}[1]
    \Procedure{Algorithm2}{$\lr{P_{|\Xv|}}$} \vspace*{0.35ex}
    \State {$\varepsilon_2$} \hskip2ex \Comment{Desired precision in Capacity estimation}
    \State $\K \gets \mleft|\mathsf{supp} \lr{P_{|\Xv|}}\mright| $
    \Repeat 
    \State \Comment{Evaluate PMF probabilities and Lower Bound} \vspace*{0.2ex}
    \State {$\lrs{P_{|\Xv|},\underline{\C}(\A)} \gets \textsc{\textrm{Algorithm1}}\lr{P_{|\Xv|}}$} \vspace*{0.35ex}
    \If {$i \lr{ 0;P_{|\Xv|} } \ge i \lr{ \rho;P_{|\Xv|} }, \ \forall \rho \in (0,\A] $}
    \vspace*{.45ex}
        \State \Comment{Add new mass point in 0 and reinitialize PMF}
        \State $\K \gets \K+1$
        \State $\rho_\K \gets 0$
        \State $P_{|\Xv|}\lr{\rho_i} \gets 1/\K , \ \forall i=1,\dots,\K$ \vspace*{0.35ex}
        \State \Comment{Update PMF probabilities and Lower Bound} 
        \State {$\lrs{P_{|\Xv|},\underline{\C}(\A)} \gets
        \textsc{\textrm{Algorithm1}}\lr{ P_{|\Xv|}}$} \vspace*{0.34ex}
        \EndIf
        \State \Comment{Evaluate the Capacity Upper Bound} \vspace*{-0.75ex}
        \State $\vcenter{\begin{flalign} \label{eq:objfunc}
        \overline{\C}(\A) \gets \max_{\rho}\: i \lr{\rho;P_{|\Xv|}}, \ \forall \rho \in (0,\A] && \hspace*{-9.75ex}
        \end{flalign}}$ \vspace*{-1ex}
        \State \Comment{By using~\eqref{eq:derivatives}, update $\lrc{\rho_j}$ as}
        \State $\vcenter{\begin{flalign}
        \rho_j \gets \sqrt{ \lr{ \rho_j }^2 + \mu \frac{\partial}{\partial \lr{ \rho_j }^2} \I}, \ \forall j=2, \dots, \K, &&
        \hspace*{-9.75ex} \end{flalign}}$
        \hspace*{5.5ex} where $\mu$ is the gradient ascend step size \vspace*{0.75ex}
    \Until{$\overline{\C}(\A) - \underline{\C}(\A) < \varepsilon_2$}
    \State $\hat{P}_{|\Xv|} \gets P_{|\Xv|}$
    \State $\hat{\C}(\A) \gets \underline{\C}(\A) \ $ \Comment{Lower bound as capacity estimate}
    \vspace*{0.5ex}
    \State \Return {$\hat{P}_{|\Xv|},\hat{\C}(\A)$}
    \EndProcedure
\end{algorithmic}
\end{algorithm}
\vspace*{-5ex}
\end{figure}
The gradient formed by the derivatives in \eqref{eq:derivatives} can be used in a gradient ascend algorithm to iteratively update $\{ \rho_i \}$. The overall procedure carried out by the outer layer is reported in Algorithm~\ref{alg:2}, which takes as input a PMF $P_{|\Xv|}$ with $\K = \underline{\K}$, $\{\rho_i\}$ evenly spaced between $0$ and $\A$, and $p_i=1/\underline{\K}$ for all $i=1,\dots,\underline{\K}$. We remark that Algorithm~\ref{alg:2} converges when the information density in~\eqref{eq:objfunc} is concave, differentiable, and its gradient is Lipschitz continuous. In the case of AWGN amplitude-constrained channels, these properties are guaranteed by~\mbox{\cite[Thm.~7]{Alg2convergence}} and~\mbox{\cite[Cor.~4]{Alg2convergence}}.

\subsection{Numerical Results} \label{sS:NumResults}
As an example, we evaluate the channel capacity for $\N$$\,=\,$$2$ and we compare it to the capacity bounds from~\cite{thangaraj2017capacity}. The computed $\hat{\C}(\A)$ with precision $\varepsilon_2$$\,=\,$$10^{-2}$ is shown in Fig.~\ref{fig:capcompThang}, where the numerical upper bound in~\cite{thangaraj2017capacity} closely matches the capacity $\hat{\C}(\A)$, within $\sim 0.1 \text{ bit per channel use (bpcu)}$.
\begin{figure}[t]
    \centering
    \definecolor{mycolor1}{rgb}{0.00000,0.44700,0.74100}%

\begin{tikzpicture}

\begin{axis}[%
width=0.85\linewidth,
height=0.7\linewidth,
at={(0pt,0pt)},
scale only axis,
xmin=-5,
xmax=30,
xlabel style={font=\color{white!15!black}},
xlabel={SNR [dB]},
ymin=0,
ymax=20,
ylabel style={font=\color{white!15!black}},
ylabel={Channel Capacity [bpcu]},
axis background/.style={fill=white},
ylabel shift = -3 pt,
xlabel shift = 0 pt,
xmajorgrids,
ymajorgrids,
legend style={font=\footnotesize, at={(0.025,0.965)}, anchor=north west, legend cell align=left, align=left, draw=white!15!black}
]
\addplot [color=black, line width=1.0pt]
  table[row sep=crcr]{%
-5	0.883429147161471\\
-3.15789473684211	1.21872675358513\\
-1.31578947368421	1.65501826082332\\
0.526315789473684	2.20229718435949\\
2.36842105263158	2.85588044939158\\
4.21052631578947	3.60117486340958\\
6.05263157894737	4.42121895853924\\
7.89473684210526	5.30039872570223\\
9.73684210526316	6.22557392935572\\
11.5789473684211	7.18615614522837\\
13.4210526315789	8.17467955820727\\
15.2631578947368	9.19720557844304\\
17.1052631578947	10.2533970698507\\
18.9473684210526	11.3373197567778\\
20.7894736842105	12.4461269643108\\
22.6315789473684	13.573949758145\\
24.4736842105263	14.7202799192966\\
26.3157894736842	15.8797882808005\\
28.1578947368421	17.0507618901881\\
30	18.2324422259277\\
};
\addlegendentry{ Upper Bound from \cite[Eq. (82)]{thangaraj2017capacity}}

\addplot [color=red, dashed, line width=1.5pt]
  table[row sep=crcr]{%
-5	0.791910394266937\\
-3.15789473684211	1.13430469361016\\
-1.31578947368421	1.58516432659703\\
0.526315789473684	2.14978525160798\\
2.36842105263158	2.81880874677174\\
4.21052631578947	3.57059542786861\\
6.05263157894737	4.38379432560111\\
7.89473684210526	5.25090441484177\\
9.73684210526316	6.1659578723744\\
11.5789473684211	7.12518773342173\\
13.4210526315789	8.12382959677494\\
15.2631578947368	9.15749531074436\\
17.1052631578947	10.2218011791161\\
18.9473684210526	11.3125947949053\\
20.7894736842105	12.4260133201573\\
22.6315789473684	13.5586575839393\\
24.4736842105263	14.7074130949544\\
26.3157894736842	15.8697057507126\\
28.1578947368421	17.0431611256071\\
30	18.2258225329622\\
};
\addlegendentry{ $\overline{\C}(\A)$}

\addplot [color=blue, line width=1.0pt, only marks, mark size=1.5pt, mark=*, mark options={solid, mycolor1}]
  table[row sep=crcr]{%
-5	0.791910394266937\\
-3.15789473684211	1.13430469361016\\
-1.31578947368421	1.58516432659703\\
0.526315789473684	2.14978525160798\\
2.36842105263158	2.81880874677174\\
4.21052631578947	3.57059542786861\\
6.05263157894737	4.38379431776147\\
7.89473684210526	5.24949062389456\\
9.73684210526316	6.16457251274663\\
11.5789473684211	7.12375632287035\\
13.4210526315789	8.12241573717304\\
15.2631578947368	9.15607827247074\\
17.1052631578947	10.2203729358227\\
18.9473684210526	11.3111587586078\\
20.7894736842105	12.4245991732913\\
22.6315789473684	13.557232925907\\
24.4736842105263	14.7060114296816\\
26.3157894736842	15.8682641145174\\
28.1578947368421	17.0417185740284\\
30	18.2243855881108\\
};
\addlegendentry{ $\underline{\C}(\A)$}

\addplot [color=mycolor1, dotted, line width=1.5pt]
  table[row sep=crcr]{%
-5	0.43947253822714\\
-3.15789473684211	0.647173806504281\\
-1.31578947368421	0.938260116574028\\
0.526315789473684	1.33313021264311\\
2.36842105263158	1.84828827749303\\
4.21052631578947	2.49189689555953\\
6.05263157894737	3.26122654373716\\
7.89473684210526	4.14353390971204\\
9.73684210526316	5.11982237576971\\
11.5789473684211	6.16938066538407\\
13.4210526315789	7.27321668881948\\
15.2631578947368	8.41575350927211\\
17.1052631578947	9.58513675099946\\
18.9473684210526	10.7727893732969\\
20.7894736842105	11.9727126333063\\
22.6315789473684	13.18080498219\\
24.4736842105263	14.3943038794849\\
26.3157894736842	15.611367039731\\
28.1578947368421	16.8307738870142\\
30	18.0517192173295\\
};
\addlegendentry{ Lower Bound from \cite[Eq. (78)]{thangaraj2017capacity}}

\end{axis}

\end{tikzpicture}%
    \caption{Channel Capacity bounds and estimate for $\N=2$, with tolerance $\varepsilon_2 = 10^{-2}$.}
    \label{fig:capcompThang}
\end{figure}
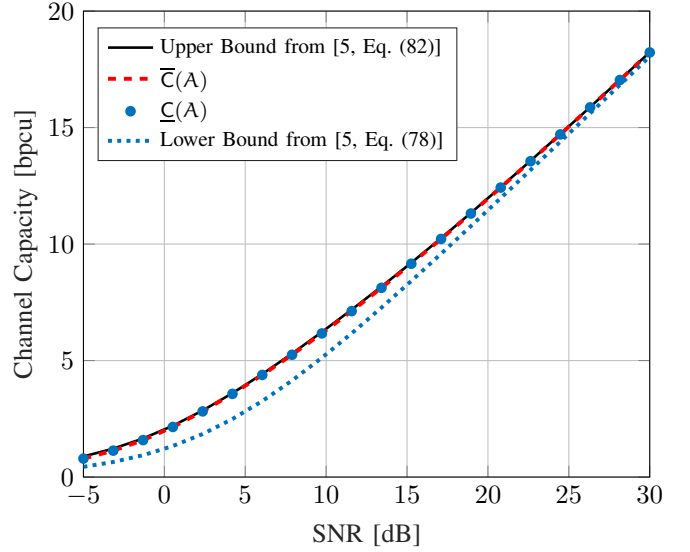
\begin{figure}[t]
    \centering
    \input{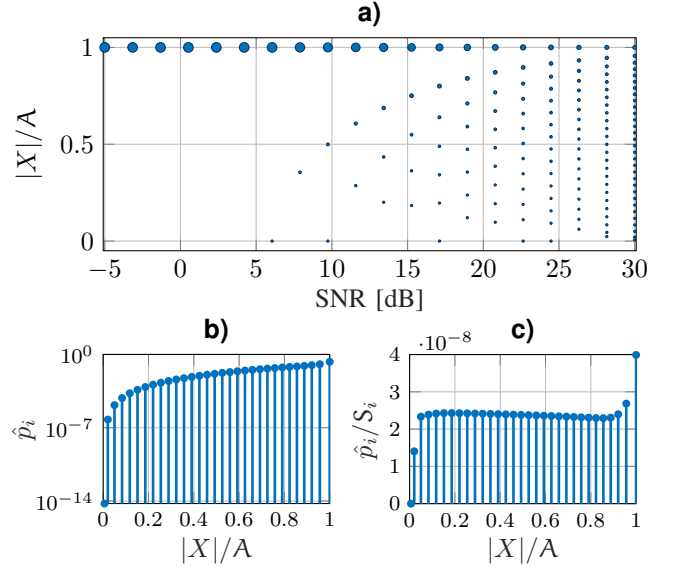}
    \caption{For $\N=2$ and tolerance $\varepsilon_2 = 10^{-2}$: \textsf{\textbf{a)}} Evolution of the normalized radii $\{\hat{\rho}_i/\A\}$ vs SNR, \textsf{\textbf{b)}} Estimated PMF $\hat{P}_{|\Xv|}$ at $\text{SNR} = 30 \text{ dB} $, and \textsf{\textbf{c)}} PDF $\hat{P}_{|\Xv|}/\Ss_i$ at $\text{SNR} = 30 \text{ dB} $.}
    \label{fig:PMFev30dB}
\end{figure}
The estimated optimal PMF $\hat{P}_{|\Xv|}$ is characterized by $\hat{\K}$, $\{\hat{p}_i\}$, and $\{\hat{\rho}_i\}$. In Fig.~\ref{fig:PMFev30dB}, the input amplitudes $|\Xv|$ are normalized to $|\Xv|/\A$ for the sake of compact illustration. In Fig.~\ref{fig:PMFev30dB}a, we show the evolution of $\{\hat{\rho}_i/\A\}$ versus the SNR, where the area of the circles is proportional to the associated $\hat{p}_i$. In Fig.~\ref{fig:PMFev30dB}b, we show the optimal PMF estimate $\hat{P}_{|\Xv|}$, for $\text{SNR}$$\,=\,$$30 \text{ dB}$. We notice that the PMF mass points tend to be evenly spaced as the SNR increases. Finally, in Fig.~\ref{fig:PMFev30dB}c we show the PDF $\hat{P}_{|\Xv|}/\Ss_i$ of a point on the $i$th hyper-sphere, again for $\text{SNR}$$\,=\,$$30 \text{ dB}$. We notice that, as the SNR increases, this PDF tends to become uniformly distributed, aside from some edge effects due to the constraint $|\Xv|<\A$.

\section{Conclusion} \label{S:conclusion}
We derived fundamental properties of the capacity-achieving input distribution for the multiple-input multiple-output additive white Gaussian noise channels subject to peak amplitude constraint. By using these insights, we also proposed an iterative procedure to numerically evaluate the optimal input distribution and consequently the information channel capacity.

\bibliographystyle{IEEEtran}
\bibliography{bibliofile}

\onecolumn
\section*{Appendix}
\begin{proof}
	Given two values $\rho_1,\rho_2$ with $\rho_1>\rho_2$, write
	\begin{align}
	i(\rho_1;F_{|\X|})-i(\rho_2;F_{|\X|}) &= \int_{0}^{\infty} \mleft(f_{\chi^2_{2\N}(\rho_1^2)}(y)-f_{\chi^2_{2\N}(\rho_2^2)}(y)\mright)\log\frac{y^{\N-1}}{f_{|\Y|^2}(y;F_{|\X|})} dy \\
	&=-\int_{0}^{\infty} \mleft(F_{\chi^2_{2\N}(\rho_2^2)}(y)-F_{\chi^2_{2\N}(\rho_1^2)}(y)\mright)\frac{d}{dy}\log\frac{f_{|\Y|^2}(y;F_{|\X|})}{y^{\N-1}} dy \label{eq:der1}
	\end{align}
	where we have integrated by parts. Now notice that
	\begin{equation}
	\int_{0}^{\infty} \mleft(F_{\chi^2_{2\N}(\rho_2^2)}(y)-F_{\chi^2_{2\N}(\rho_1^2)}(y)\mright) dy = \rho_1^2 - \rho_2^2, \label{eq:auxpdf}
	\end{equation}
	and the integrand function is always positive, being $\chi^2_{2\N}(\rho_1^2)$ statistically dominant compared to $\chi^2_{2\N}(\rho_2^2)$. Since $\chi^2_{2\N}(\rho_1^2)$ statistically dominates $\chi^2_{2\N}(\rho_2^2)$, the integrand function in \eqref{eq:auxpdf} is always positive. We can introduce an auxiliary output random variable $Q$ with PDF
	\begin{equation}\label{eq:def_fQ}
	f_Q(y;\rho_1,\rho_2) = \frac{F_{\chi^2_{2\N}(\rho_2^2)}(y)-F_{\chi^2_{2\N}(\rho_1^2)}(y)}{\rho_1^2 - \rho_2^2}, \qquad y>0
	\end{equation} 
	to rewrite \eqref{eq:der1} as follows:
	\begin{equation}
	i(\rho_1;F_{|\X|})-i(\rho_2;F_{|\X|}) = -(\rho_1^2-\rho_2^2)\int_{0}^{\infty} f_Q(y;\rho_1,\rho_2)\frac{d}{dy}\log\frac{f_{|\Y|^2}(y;F_{|\X|})}{y^{\N-1}} dy.\label{eq:afterQ}
	\end{equation}
	We evaluate the derivative in \eqref{eq:afterQ} as:
	\begin{align}
	\frac{d}{dy}\log\frac{f_{|\Y|^2}(y;F_{|\X|})}{y^{\N-1}} &\stackrel{(a)}{=} \frac{y^{\N-1}}{f_{|\Y|^2}(y;F_{|\X|})}\int_{0}^{\A} \frac{d}{dy} \frac{f_{\chi^2_{2\N}(\rho^2)}(y)}{y^{\N-1}} dF_{|\X|}(\rho) \\
	&\stackrel{(b)}{=}\frac{y^{\N-1}}{f_{|\Y|^2}(y;F_{|\X|})}\int_{0}^{\A} \mleft(\frac{f_{\chi^2_{2(\N-1)}(\rho^2)}(y)}{2y^{\N-1}}-(\frac{1}{2}+\frac{\N-1}{y})\frac{f_{\chi^2_{2\N}(\rho^2)}(y)}{y^{\N-1}}\mright) dF_{|\X|}(\rho) \\
	&=\expcnd{\frac{1}{2} \frac{f_{\chi^2_{2(\N-1)}(|\X|^2)}(|\Y|^2)}{f_{\chi^2_{2\N}(|\X|^2)}(|\Y|^2)}-(\frac{1}{2}+\frac{\N-1}{|\Y|^2}) }{|\Y|^2=y} \\
	&=\expcnd{\frac{1}{2}\frac{|\X|}{|\Y|} \frac{I_{\N-2}(|\X| |\Y|)}{I_{\N-1}(|\X| |\Y|)}-(\frac{1}{2}+\frac{\N-1}{|\Y|^2}) }{|\Y|^2=y} \label{eq:derlog}
	\end{align} 
	where in step $(a)$ we used $f_{|\Y|^2}(y;F_{|\X|}) = \int_{0}^{\A} f_{\chi^2_{2\N}(\rho^2)}(y) dF_{|\X|}(\rho)$ and in step $(b)$ the relationship
	\begin{equation}
	\frac{d}{dy} f_{\chi^2_{2\N}(\rho^2)}(y) = \frac{1}{2} f_{\chi^2_{2(\N-1)}(\rho^2)}(y)-\frac{1}{2} f_{\chi^2_{2\N}(\rho^2)}(y).
	\end{equation}
	 Putting together \eqref{eq:afterQ} and \eqref{eq:derlog} we get
	\begin{equation}
	i(\rho_1;F_{|\X|})-i(\rho_2;F_{|\X|}) = -(\rho_1^2-\rho_2^2) \expect{\expcnd{\frac{1}{2}\frac{|\X|}{|\Y|} \frac{I_{\N-2}(|\X| |\Y|)}{I_{\N-1}(|\X| |\Y|)}-(\frac{1}{2}+\frac{\N-1}{|\Y|^2})}{|\Y|^2=Q}}.
	\end{equation}
	We are now in the position to compute the derivative of the information density as:
	\begin{align}
	i'(\rho;F_{|\X|}) &= \lim_{h\rightarrow 0} \frac{i(\rho+h;F_{|\X|})-i(\rho;F_{|\X|})}{h} \\
	&=-2\rho\: \expect{\expcnd{\frac{1}{2}\frac{|\X|}{|\Y|} \frac{I_{\N-2}(|\X| |\Y|)}{I_{\N-1}(|\X| |\Y|)}-(\frac{1}{2}+\frac{\N-1}{|\Y|^2}) }{|\Y|^2=Q'}}
	\end{align}
	where $Q'\sim \chi^2_{2(\N+1)}(\rho^2)$ thanks to the following Lemma.
	\end{proof}
\begin{lem}
	We have
	\begin{equation}
	\lim_{h\rightarrow 0} f_Q(y;\rho+h,\rho) = f_{\chi^2_{2(\N+1)}(\rho^2)}(y), \qquad y>0.
	\end{equation}
\end{lem}
\begin{proof}
	Thanks to the definition \eqref{eq:def_fQ}, we have
	\begin{align}
	\lim_{h\rightarrow 0} f_Q(y;\rho+h,\rho) &= \lim_{h\rightarrow 0} \frac{F_{\chi^2_{2\N}(\rho^2)}(y)-F_{\chi^2_{2\N}((\rho+h)^2)}(y)}{h(2\rho+h)} \\
	&=\lim_{h\rightarrow 0} \frac{1}{h(2\rho+h)} \int_{0}^{y} \mleft(f_{\chi^2_{2\N}(\rho^2)}(t)-f_{\chi^2_{2\N}((\rho+h)^2)}(t)  \mright) dt \\
	&= \frac{1}{2\rho}\int_{0}^{y} \sum_{i=0}^{\infty} \lim_{h\rightarrow 0} \frac{1}{h}\mleft(\frac{e^{-\rho^2/2}(\rho^2/2)^i}{i!}-\frac{e^{-(\rho+h)^2/2}((\rho+h)^2/2)^i}{i!} \mright) f_{\chi^2_{2\N+2i}}(t) dt \\
	&= \frac{1}{2\rho}\int_{0}^{y} \sum_{i=0}^{\infty} \frac{d}{d\rho} \mleft(\frac{e^{-\rho^2/2}(\rho^2/2)^i}{i!}\mright) f_{\chi^2_{2\N+2i}}(t) dt \\
	&= \frac{1}{2}\int_{0}^{y} \sum_{i=0}^{\infty}  \mleft(-\frac{e^{-\rho^2/2}(\rho^2/2)^i}{i!}+\frac{e^{-\rho^2/2}(\rho^2/2)^{i-1}}{(i-1)!}\mathbbm{1}(i\ge 1)\mright) f_{\chi^2_{2\N+2i}}(t) dt	\\
	&= \frac{1}{2}\int_{0}^{y}  \mleft(-f_{\chi^2_{2\N}(\rho^2)}(t)+f_{\chi^2_{2(\N+1)}(\rho^2)}(t)\mright)  dt \\
	&=\int_{0}^{y}  \frac{d}{dt}f_{\chi^2_{2(\N+1)}(\rho^2)}(t)  dt \\
	&=f_{\chi^2_{2(\N+1)}(\rho^2)}(y),
	\end{align}
	where $\mathbbm{1}(\cdot)$ is the indicator function.
	\end{proof}

\end{document}